\documentclass[conference, 10pt, letterpaper]{ieeeconf}
\IEEEoverridecommandlockouts
\pdfoutput=1

\makeatletter

\let\proof\@undefined
\let\endproof\@undefined
\makeatother

\usepackage{algorithm}
\usepackage{amsmath,amssymb,amsthm}
\usepackage[font=small]{caption}
\usepackage[font=footnotesize]{subcaption}
\usepackage{booktabs}
\usepackage{multirow}
\newtheorem{theorem}{Theorem}[section]
\newtheorem{proposition}[theorem]{Proposition}
\newtheorem{corollary}[theorem]{Corollary}
\newtheorem{lemma}[theorem]{Lemma}

\theoremstyle{definition}

\newtheorem{problem}[theorem]{Problem}

\theoremstyle{remark}
\newtheorem{remark}[theorem]{Remark}

\usepackage{epstopdf}
\usepackage{graphicx,url}
\usepackage{epstopdf}
\usepackage{color}
\usepackage{float}
\usepackage{comment}
\usepackage{mathtools}

\renewcommand{\thesection}{\Roman{section}} 
\renewcommand{\thesubsection}{\thesection.\Roman{subsection}}
\usepackage{algpseudocode}

\newcommand\oprocendsymbol{\hbox{$\bullet$}}
\newcommand\oprocend{\relax\ifmmode\else\unskip\hfill\fi\oprocendsymbol}

\urldef{\smith}\url{stephen.smith@uwaterloo.ca}
\urldef{\armin}\url{a6sadeghiyengejeh@uwaterloo.ca}

\renewcommand{\baselinestretch}{0.985}

\title{\Large \textbf{On Efficient Computation of Shortest Dubins Paths \\Through Three Consecutive Points}}
\author{Armin Sadeghi \qquad Stephen L. Smith \thanks{This research is partially
    supported by the Natural Sciences and Engineering Research Council
    of Canada (NSERC). }
  \thanks{The authors are with the Department of Electrical and
    Computer Engineering, University of Waterloo, Waterloo ON, N2L 3G1
    Canada  (\armin; \smith) }}

\begin{document}
\maketitle

\begin{abstract}

In this paper, we address the problem of computing optimal paths through three consecutive points for the curvature-constrained forward moving Dubins vehicle. Given initial and final configurations of the Dubins vehicle, and a midpoint with an unconstrained heading, the objective is to compute the midpoint heading that minimizes the total Dubins path length. We provide a novel geometrical analysis of the optimal path, and establish new properties of the optimal Dubins' path through three points. We then show how our method can be used to quickly refine Dubins TSP tours produced using state-of-the-art techniques.  We also provide extensive simulation results showing the improvement of the proposed approach in both runtime and solution quality over the conventional method of uniform discretization of the heading at the mid-point, followed by solving the minimum Dubins path for each discrete heading.
\end{abstract}

\section{Introduction}\label{sec:intro}
Routing problems for non-holonomic vehicles have been studied extensively in the fields of robotics and autonomous systems~\cite{isaacs2011algorithms,tang2005motion,isaiah2015motion, herisse2013shortest,bullo2011dynamic}. The non-holonomic motion of a forward-moving Dubins vehicle with bounded turning radius \cite{dubins1957curves} is commonly studied as a model for fixed-wing aerial vehicles. A configuration of a Dubins vehicle consists of a location $(x,y)$ in the Euclidean plane and a heading $\alpha \in [0, 2\pi)$. The motion of the Dubins' vehicle with minimum-turning radius $R_{\min}$ and control input $u \in \left[-1/R_{\min}, 1/R_{\min}\right]$ is governed by the following equations: 
\[
\dot x = \cos\alpha, \quad \dot y = \sin\alpha, \quad \dot \alpha = u. 
\]

Dubins \cite{dubins1957curves} provided the set of candidate optimal paths between pair-wise configurations of the Dubins vehicle.

In this paper, we focus on the Dubins path problem between three consecutive points, where headings at only the initial and final point are fixed. Our interest in this problem stems from two applications.  First, given a Dubins path through a set of points, a fast solution to this problem provides a method for inserting a new point into the Dubins path with minimum additional cost. Second, we show how it can be used as a tool to perform repeated local optimizations on a Dubins path through a set of points.

\emph{Related work:}  
Ma \emph{et al.} \cite{ma2006receding} study the optimal Dubins paths for three consecutive points where the initial heading is fixed and the midpoint and final point have free headings, building on the optimal control results in~\cite{sussmann1991shortest}. Under the assumption that the pairwise Euclidean distance between all points is at least $2 R_{\min}$, the authors provide a sufficient condition for the optimal path between the points. In addition, a receding horizon algorithm is proposed to construct feasible Dubins path on an ordered set of points. Although the optimal control approaches provide useful optimality conditions on the Dubins paths through three points, they do not directly provide an efficient method of finding the path that satisfies the conditions.

The authors of \cite{goaoc2013bounded} formulate a family of convex optimization 
sub-problems to address the problem of the optimal Dubins path between a set of $n$ ordered points with distance at least $4 R_{\min}$ apart. The drawback of the approach is that the number of convex optimization sub-problems can grow to $2^{(n-2)}$ in the worst case. Their approach provides a solution to the three-point Dubins problem, but it requires solving several convex optimization problems. A heuristic was recently proposed~\cite{vana2015dubins} to extend this method to the problem of Dubins paths through neighborhoods.

Another closely related problem is the Dubins TSP, where given $n$ points, the objective is to sequence the points and choose a heading at each point such that the resulting Dubins tour length is minimum. In \cite{savla2008traveling,rathinam2007resource,macharet2014orientation}, approximation algorithms are proposed to assign headings to the points given the optimal ordering of the Euclidean TSP problem on the same set of points. In \cite{macharet2014orientation}, the headings are assigned by a heuristic solution to the three-point Dubins problem considered in this paper. This heuristic approach is adopted in \cite{macharet2013efficient} to insert points into the tours of multiple-Dubins vehicles.  

In \cite{le2012dubins}, the continuous interval of headings at each point is approximated by a finite number of samples. Each sample, along with the position of the corresponding point forms a configuration, and the problem reduces to computing a generalized traveling salesman problem (GTSP) tour that visits one configuration for each point.  The authors in \cite{cons2014integrating} present an experimental comparison of Dubins TSP algorithms including the GTSP approach. 
Recently and built on the results for the pairwise optimal Dubins interval path, Manyam and Rathinam~\cite{manyam2015tightly} proposed a Dubins TSP algorithm based on uniform discretization of the headings at each point to intervals. 

\emph{Contributions:} The focus of this paper is to provide an efficient method for computing the optimal Dubins path between three consecutive points.  We present a novel analysis of the problem that relies on inversive geometry, and results in a set of equations defining the optimal heading at the mid-point. We  provide a simple method to approximate the optimal heading, and give bounds on its worst-case deviation from optimal. We then present an iterative method that is guaranteed to converge to the optimal solution. In simulation, we compare our approach to the uniform discretization method of~\cite{le2012dubins} in both solution quality and computation time. Finally, we show that a Dubins TSP can be solved using a coarse heading discretization followed by repeated heading optimization using our technique to achieve high-quality tours in approximately $8\%$ of the computation time.

The paper is organized as follows.  In Section~\ref{sec:background}, we provide background on inversive geometry. In Section~\ref{sec:Formulation}, we formulate the Dubins path problem through three consecutive points. Section~\ref{sec:inversion_proof} derives equations for the optimal path through three points. In Section~\ref{sec:Optimal CSCSC Paths} we present an approximation method for the heading at the mid-point and an exact algorithm to obtain the optimal heading. In Section~\ref{sec:simulations}, we provide applications and benchmarking results.

\section{Preliminaries} \label{sec:background}
Here we provide a brief background on circle inversion~\cite{henderson2000experiencing}.  In two dimensional geometry, circle inversion is a mapping of a geometric object $Q$  with respect to a circle $\mathcal{C} = \mathrm{circle}(O, R)$ to another object $\mathrm{inv}(Q, \mathcal{C})$. The inverse of a point $P$ with respect to $\mathcal{C}$ is a point on the segment $\overline{OP}$ with distance $\frac{R^2}{|OP|}$ from $P$ (see Figure~\ref{fig:inversion_def}). The inverse of a line (resp. circle) with respect to circle $\mathcal{C}$ is a circle, unless the line (resp. circle) contains $O$, in which case the inverse is a line. The inverse of a line (resp. circle) is obtained by inverting three points on the line (resp. circle).  With a slight abuse of terminology, we define inverse of a line segment $S$ with respect to $\mathcal{C}$ to be the inverse of the infinite line containing the line segment $S$ with respect to  $\mathcal{C}$.

The angle between a circle and an intersecting line is defined as the angle between the line and the tangent to the circle at the intersection point. The angles between the intersecting lines and circles are preserved under the circle inversion operation.

\begin{figure}
\centering
        \includegraphics[width=0.55\linewidth, keepaspectratio=true]{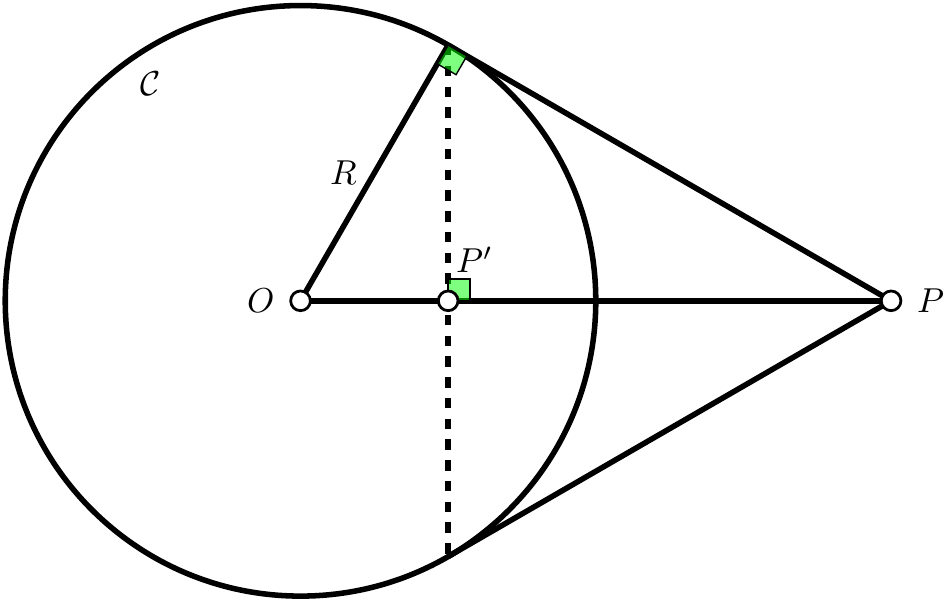}
        \caption{Point $P'$ is the inverse of point $P$ with respect to circle $\mathcal{C} = \mathrm{circle}(O, R)$.}
        \label{fig:inversion_def}
\end{figure}

\section{ Problem Formulation} 
\label{sec:Formulation}
We now formulate the problem of finding an optimal path for a Dubins vehicle between three consecutive points. 

\subsection{Three-Point Dubins Path}

Let the tuple $X_i = (\mathbf{x}_i, \alpha_i)$ denote a Dubins vehicle configuration, consisting of a point $\mathbf{x}_i$ in the Euclidean plane, and a heading $\alpha_i \in [0,2\pi)$ at $\mathbf{x}_i$. An alternative representation of the heading at $\mathbf{x}_i$ is two circular arcs (left and right turns) containing $\mathbf{x}_i$ and tangent to the heading. Given initial and final configurations $X_i$ and $X_f$, along with a midpoint $\mathbf{x}_m$ with free heading, the three-point problem is defined by the tuple $(X_i, \mathbf{x}_m, X_f)$ and stated as follows.

\begin{problem}[Three-point Dubins path]  \label{prob:three_point}
Given a tuple $(X_i, \mathbf{x}_m, X_f)$, with pairwise Euclidean distances between the points $\mathbf{x}_i$, $\mathbf{x}_m$, $\mathbf{x}_f$ of at least $4R_{\min}$, find a heading $\alpha_m$ at $\mathbf{x}_m$ such that the length of an optimal Dubins path starting at $X_i$, passing through $X_m = (\mathbf{x}_m, \alpha_m)$, and arriving at $X_f$ is minimum.
\end{problem}

From the Bellman's principle of optimality \cite{bellman1956dynamic}, the optimal Dubins path through three configurations is obtained by concatenating two optimal Dubins paths between the  pairs. Given two configurations $X_1$ and $X_2$, the optimal Dubins path from $X_1$ to $X_2$ can be computed in constant time \cite{dubins1957curves}. 
The optimal Dubins paths between two configurations is in the set $\{CCC, CSC\}$ where $S$ is a straight line segment and $C$ is a circular turn with minimum turning radius in either left $L$ or right $R$ direction. Therefore, in general the optimal path through three points is obtained by concatenating two Dubins paths as follows.
\begin{align*}
\{(C_1C_2C_3)_1(C_4C_5C_6)_2,(C_1C_2C_3)_1(C_4S_5C_6)_2,\\(C_1S_2C_3)_1(C_4C_5C_6)_2,(C_1S_2C_3)_1(C_4S_5C_6)_2\}.
\end{align*}
From \cite{goaoc2013bounded} the set of optimal Dubins paths under $4R_{\min}$ distance assumption of Problem \ref{prob:three_point} is reduced to $(C_1S_2C_3)_1(C_4S_5C_6)_2$ . The $4R_{\min}$ distance constraint is relaxed further in Section \ref{sec:simulations}.

\subsection{Properties of Three-Point Dubins Path}
In a path of type $(C_1S_2C_3)_1(C_4S_5C_6)_2$, the arc segments $C_3$ and $C_4$ are the two incident path segments to the mid-point. In the optimal solution to Problem \ref{prob:three_point}, the two arcs incident to the mid-point have equal lengths and both are in the same turning direction i.e., left turn or right turn \cite{goaoc2013bounded}. Thus for simplicity we represent the path as  $C_1S_2C_3S_4C_5$. We summarize the properties of the optimal Dubins path through three consecutive points, provided in \cite{goaoc2013bounded}, as follows.

\begin{lemma}[Three-point Dubins] \label{lem:bisect}
Given $(X_i, \mathbf{x}_m, X_f)$, in a shortest path of type $C_1S_2C_3S_4C_5$, the line segment between $\mathbf{x}_m$ and the center of the circle associated with the optimal heading bisects the angle between the line segments $S_2$ and $S_4$.
\end{lemma}
Substituting the left $L$ and right $R$ turns for each $C_i$ in the path $C_1S_2C_3S_4C_5$, we obtain the set of 8 candidate optimal path types for Problem \ref{prob:three_point}.

In \cite{ma2006receding}, the authors address a variation of Problem \ref{prob:three_point} in which the final heading is also free, i.e., the problem $(X_i, \mathbf{x}_m, \mathbf{x}_f)$.  The authors show that under a $2R_{\min}$ distance constraint the optimal path is of type $C_1S_2C_3S_4$, and the optimal heading bisects the angle between $S_2$ and $S_4$ as in Lemma \ref{lem:bisect}. Limiting the path types in the problem  $(X_i, \mathbf{x}_m, X_f)$ to $C_1S_2C_3S_4C_5$,  the result of Lemma \ref{lem:bisect} applies even without the $4R_{\min}$ constraint. The proof follows directly from the proof in \cite{ma2006receding} for the problem instance $(X_i, \mathbf{x}_m, \mathbf{x}_f)$.

\section{Optimal Path and Inversive Geometry}
\label{sec:inversion_proof}
In this section we use inversive geometry to establish properties of optimal paths of type $C_1S_2C_3S_4C_5$, that form the basis of our solution approach to Problem~\ref{prob:three_point}.

\subsection{Inversive Geometry in Dubins Paths}
 Figure \ref{fig:inversion_first} shows the optimal path for the case $R_1S_2R_3S_4L_5$. The points $A,B$ and $\mathbf{x}_c$ are the centers of the circles associated with the headings at the points $\mathbf{x}_i, \mathbf{x}_f$ and $\mathbf{x}_m$ respectively. In Figure \ref{fig:inversion_first}, the common tangent of the circles centered at $A$ and $\mathbf{x}_c$ is an \emph{outer-common tangent} and the common tangent of the circles centered at $\mathbf{x}_c$ and $B$ is an \emph{inner-common tangent}.
 
Figure~\ref{fig:inversion_th} shows the inverse of the components of the path with respect to the circle $\mathcal{C}$ centered at $\mathbf{x}_m$ with radius $R_{\min}$. Each $S$ segment in Figure \ref{fig:inversion_first} is shown as a line in Figure \ref{fig:inversion_th}. The circle inversion operation on each line generates a circle containing the mid-point $\mathbf{x}_m$, shown in Figure \ref{fig:inversion_th} in the same color. The inverse of the circle associated with the heading at $\mathbf{x}_m$ is a line passing through the two intersection points of $\mathrm{circle}(\mathbf{x}_m, R_{\min})$ and $\mathrm{circle}(\mathbf{x}_c, R_{\min})$.

\begin{figure}
\centering
        \includegraphics[width=0.55\linewidth, keepaspectratio=true]{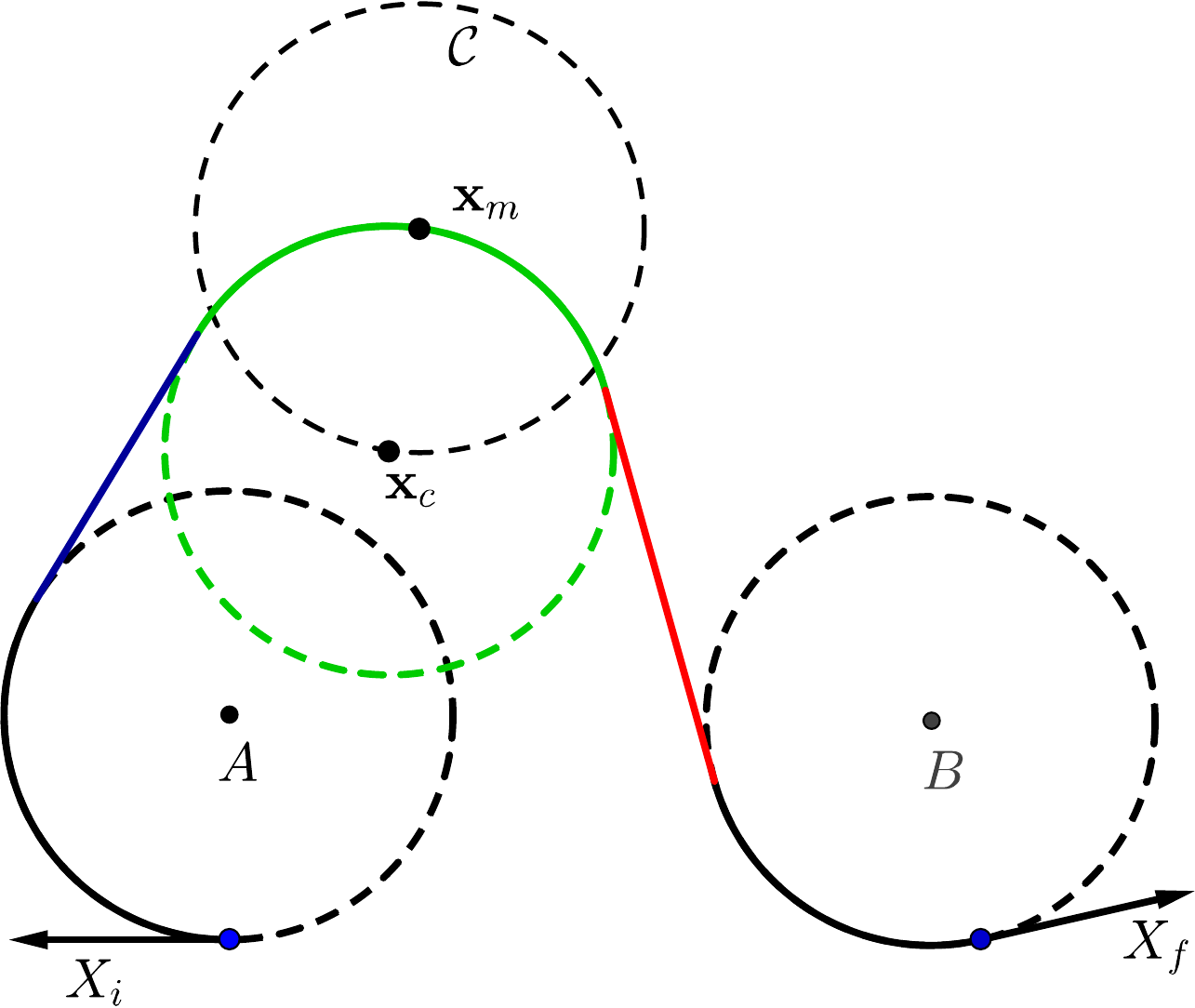}
        \caption{An optimal path of type $R_1S_2R_3S_4L_5$. Each component of the path is sketched
in different colors.}
        \label{fig:inversion_first}
\end{figure}
\begin{figure}
\centering
        \includegraphics[width=0.55\linewidth, keepaspectratio=true]{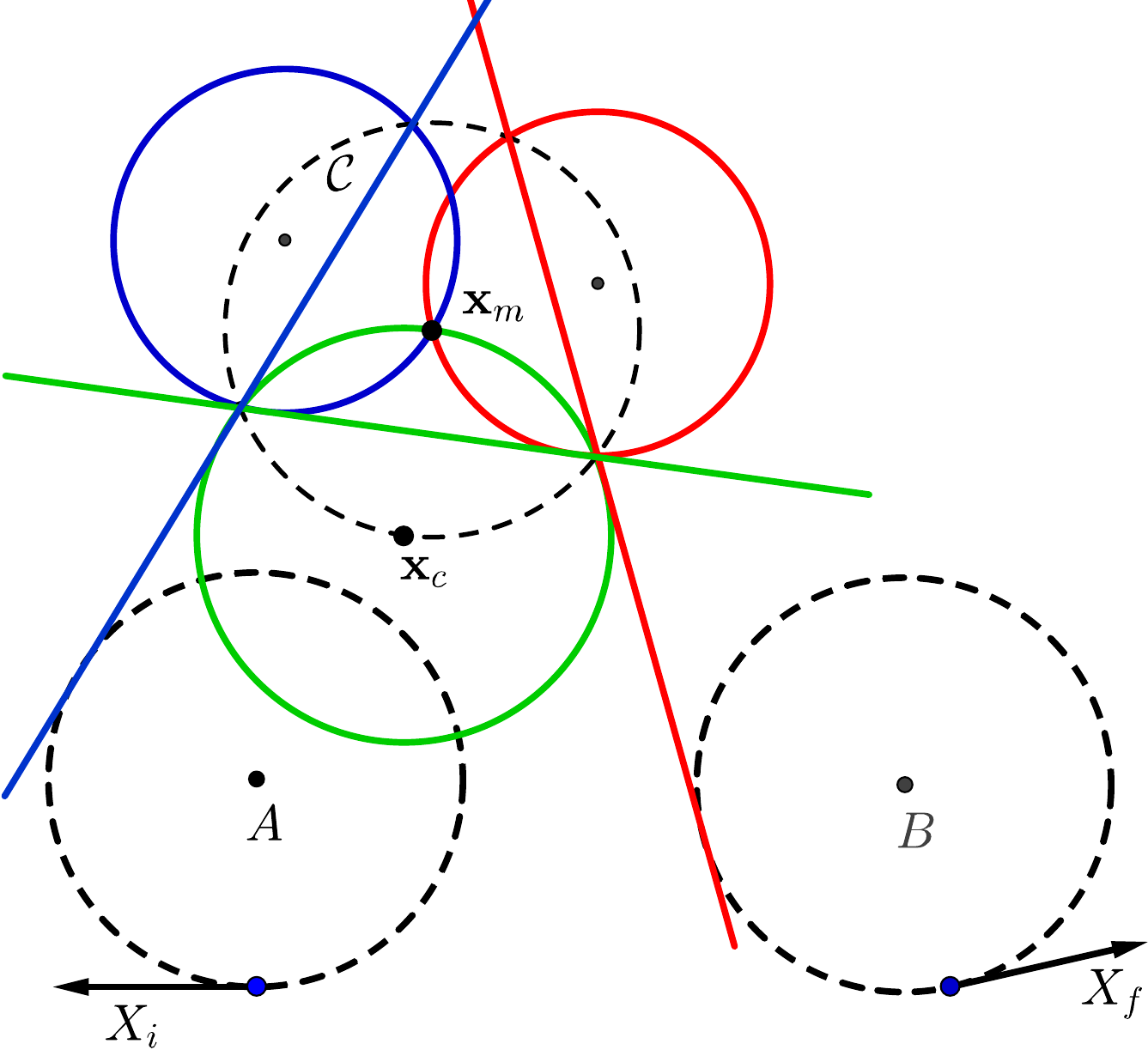}
        \caption{Inverse of the components of the path $R_1S_2R_3S_4L_5$ with respect to the
dashed circle $\mathcal{C} = \mathrm{circle}(\mathbf{x}_m, R_{\min})$.}
         \label{fig:inversion_th}
\end{figure}

The following lemma provides a sufficient condition for optimality of a $C_1S_2C_3S_4C_5$ path based on Lemma~\ref{lem:bisect}.
\begin{lemma}[Radius of inverted circles in an optimal path]
\label{lem:Radius_inverted_circles}
In any optimal path of type $C_1S_2C_3S_4C_5$, the inverses of the line segments,
$S_2$ and $S_4$, with respect to a circle centered at $\mathbf{x}_m$ with radius $R_{\min}$ are two circles of equal radius.
\end{lemma}
\begin{proof}
Consider any optimal path of type $C_1S_2C_3S_4C_5$ --- for such a path we can define the
following quantities as shown in Figure \ref{fig:inversion_proof}. Let $P$ be the intersection of the two line
segments in the optimal path and $P'$ be the inverse of $P$ with respect to $\mathcal{C}$. Note that $P'$ is the inverse of the point $P$ and the line $\overline{PP'}$ contains
$\mathbf{x}_m$ by the definition of $P'$, thus the inverse of $\overline{PP'}$ with respect to $\mathcal{C}$ is the line itself. The point $P$ is common in both lines $S_2$ and $S_4$, thus the point $P'$ lies on both circles $\mathrm{inv}(S_2, \mathcal{C})$ and $ \mathrm{inv}(S_4, \mathcal{C})$. From Lemma \ref{lem:bisect} we know that the line $\overline{PP'}$ is the bisector of the angle between two line segments, $S_2$ and $S_4$. From Section \ref{sec:background} circle inversion preserves the angle between $\overline{PP'}$ and $S_2$ and angle between $\overline{PP'}$ and $S_4$. Therefore, the angles between the line $\overline{PP'}$ and circles  $\mathrm{inv}(S_2, \mathcal{C})$ and $ \mathrm{inv}(S_4, \mathcal{C})$ are equal. Thus we have, $\angle P'C\mathbf{x}_m = \angle P'D\mathbf{x}_m$,  which implies
\[
    \angle P'DC = \frac{\angle P'C\mathbf{x}_m}{2}=\frac{\angle P'D\mathbf{x}_m}{2}= \angle P'CD.
\]
Let $Q$ be the intersection of the lines $\overline{CD}$ and $\overline{P'\mathbf{x}_m}$. The points $\mathbf{x}_m$ and $P'$ are common in both circles $\mathrm{inv}(S_2, \mathcal{C})$ and $ \mathrm{inv}(S_4, \mathcal{C})$, implying $\angle P'QC = \angle P'QD = \frac{\pi}{2} $. Thus the triangles $\triangle P'CQ$ and $\triangle P'DQ$ are
equal given that they have common side $P'Q$ and two equal angles. Therefore, the
segments $CP'$ and $DP'$ have equal length, implying that the circles have same radius.
\end{proof}

\begin{figure}
\centering
\includegraphics[width=0.8\linewidth, keepaspectratio=true]{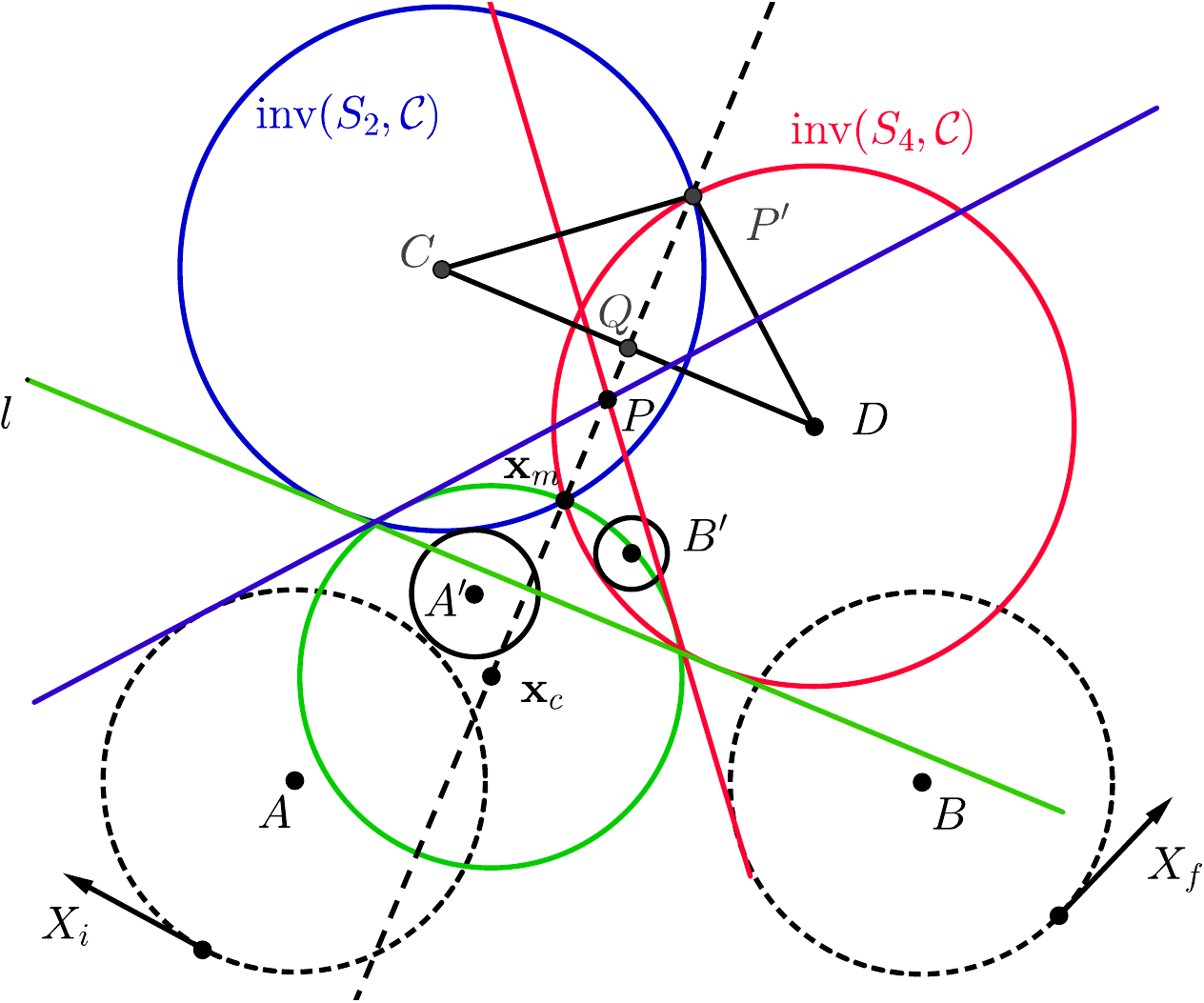}
\caption{Path $R_1S_2R_3S_4L_5$ and inverse of the path components. The optimal path given in the
figure with initial state $X_i$, final state $X_f$ and point $\mathbf{x}_m$ to visit.}
\label{fig:inversion_proof}
\end{figure}

\subsection{Optimality Condition}
Without loss of generality we set $R_{\min}$ to 1 in the rest of the paper, otherwise we scale the location of the points to satisfy the assumption. In addition, we rotate the coordinate system such that the centers $A$ and $B$ lie on the $x$-axis. Then, the optimal heading at $\mathbf{x}_m$ equals the angle between the line tangent to $\mathrm{circle}(\mathbf{x}_c, R_{\min})$ at $\mathbf{x}_m$ and the $x$-axis. 

Due to the $4R_{\min}$ distance constraint on the points, $\mathrm{circle}(A, R_{\min})$ does not contain $\mathbf{x}_m$. Therefore, the inverse of $\mathrm{circle}(A, R_{\min})$, i.e.,$\mathrm{inv}(\mathrm{circle}(A,R_{\min}), \mathcal{C})$, is a circle centered at point $A'$ with radius $r_{A'}$ (see Figure \ref{fig:inversion_proof}). The point $A'$ and radius $r_{A'}$ are defined as follows:
\begin{equation}\label{eq:inversion_eq}
r_{A'} = \frac{1}{|\overline{A\mathbf{x}_m}|^2 - 1}, \quad |\overline{A'\mathbf{x}_m}| = |\overline{A\mathbf{x}_m}|r_{A'}.
\end{equation}
Substituting $A, A'$ and $r_{A'}$ with $B, B'$ and $r_{B'}$, respectively, we can define $B'$ and $r_{B'}$.

To derive a set of equations for the optimal heading in the path $C_1S_2C_3S_4C_5$, we require the following additional definitions:
\begin{itemize}
\item $\mu_A$ is $1$ if $C_1 = C_3$ and $-1$ otherwise,
\item $\mu_B$ is $1$ if $C_5 = C_3$ and $-1$ otherwise,
\item $R$ is the radius of the circles centered at $C$ and $D$ in the optimal path,
\item $\theta 	= \angle \mathbf{x}_mCD = \angle \mathbf{x}_mDC$ (see Figure \ref{fig:inversion_proof}),
\item $\beta_1 	= \angle \mathbf{x}_mAB$ and $\beta_2 = \angle \mathbf{x}_mBA$. 
\end{itemize}
Following proposition provide the set of equations to obtain the optimal heading.
\begin{proposition}
\label{prop:inversion_proof}
The optimal heading $\alpha^*$ at $\mathbf{x}_m$ is the unique solution to the following set of equations:
\begin{align}
\frac{1}{2(\mu_A + |\overline{A\mathbf{x}_m}|\cos(\beta_1 + \theta - \alpha^*))} =R, \label{eq:trig_1}\\
\frac{1}{2(\mu_B + |\overline{B\mathbf{x}_m}|\cos(\beta_2 + \theta + \alpha^*))} =R, \label{eq:trig_2}\\
\frac{1}{2(1 - \sin(\theta))} =R.\label{eq:trig_3}
\end{align}
\end{proposition}
\begin{proof}
Figure \ref{fig:magnify_crop_triangles} shows the triangles $\triangle CA'\mathbf{x}_m$ and $\triangle DB'\mathbf{x}_m$ of Figure \ref{fig:inversion_proof}. The length of the segment $\overline{CA'}$depends on the line segment $S_2$.  If the line segment $S_2$ is an inner common tangent then $\mathrm{circle}(A', r_{A'})$ is contained in $\mathrm{circle}(C, R)$ and share a common tangent, therefore, $|\overline{CA'}|$ equals $R - r_{A'}$. On the other hand, if the line segment $S_2$ is an outer common tangent the $\mathrm{circle}(A', r_{A'})$ is tangent to $\mathrm{circle}(C, R)$ from outside and the length of the segment $|\overline{CA'}|$ is $R + r_{A'}$. The same applies to the segment $\overline{DB'}$ with respect to the segment $S_4$. The circles centered at $C$ and $D$ in Figure \ref{fig:inversion_proof} are tangent to line $l$. Therefore, the distance of the centers $C$ and $D$ from the line $l$ is $R$. Since the line $l$ passes through the intersection points of the two equally sized circles centered at $\mathbf{x}_m$ and $\mathbf{x}_c$, the distance of $\mathbf{x}_m$ from the line between $C$ and $D$ is $R - \frac{R_{\min}}{2}$ which results Equation \eqref{eq:trig_3}. Finally, the cosine law for triangles $\triangle C\mathbf{x}_mA'$ and $\triangle D\mathbf{x}_mB'$ in Figure \ref{fig:magnify_crop_triangles} result Equations \eqref{eq:trig_1} and \eqref{eq:trig_2}, respectively.
\end{proof}

\begin{figure}
\centering
\includegraphics[width=0.45\linewidth, keepaspectratio=true]{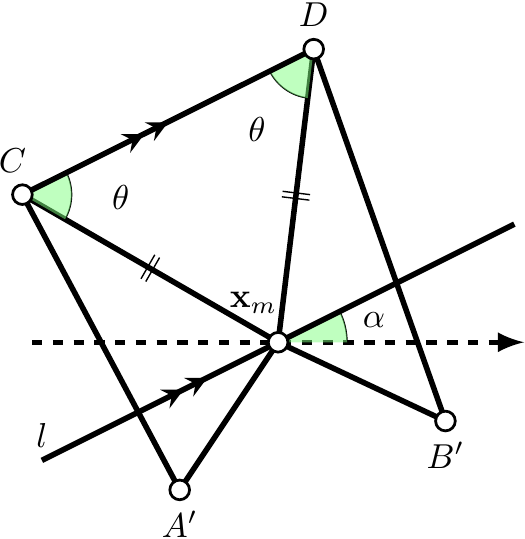}
\caption{Triangles $\triangle CA'\mathbf{x}_m$ and $\triangle DB'\mathbf{x}_m$ of Figure \ref{fig:inversion_proof}. Line $l$ contains $\mathbf{x}_m$ and is parallel to the direction of the heading at $\mathbf{x}_m$.}
\label{fig:magnify_crop_triangles}
\end{figure}

The set of unknowns in Equations \eqref{eq:trig_1}, \eqref{eq:trig_2} and \eqref{eq:trig_3} are $R$, $\alpha$ and $\theta$, where $\alpha$ is the optimal heading
at $\mathbf{x}_m$. Unfortunately, we have been unsuccessful in obtaining a closed form solution to these set of trigonometric equations. In Section \ref{sec:approx_method}, we leverage Proposition \ref{prop:inversion_proof} to bound the optimal heading at the the mid-point. Moreover, we propose a geometric method to approximate the heading, followed by an iterative procedure to converge to the optimal.

\subsection{Special Cases}

In the analysis of the inverse geometry for the $C_1S_2C_3S_4C_5$ paths, there are two special cases in which the line segments $S_2$ and $S_4$ are parallel that must be considered. First, consider the case where the lines $S_2$ and $S_4$ are parallel and non-intersecting. In this case, point $P$ is defined to lie at infinity, and thus, by the definition of the circle inversion, the point $P'$ is placed on $\mathbf{x}_m$. Note that circles $\mathrm{inv}(S_2, \mathcal{C})$ and $\mathrm{inv}(S_4, \mathcal{C})$ coincide only at $\mathbf{x}_m$ and $P'$. Since the two lines are parallel and the radius of the circle of inversion $\mathcal{C}$ is equal to the minimum turning radius, then in the optimal case the two lines are also tangent to the circle of inversion. Therefore, the inverse of the lines are equal in the radius, which yields the result of Lemma~\ref{lem:Radius_inverted_circles}. Moreover, all the parameters in Equations~\eqref{eq:trig_1}, \eqref{eq:trig_2} and \eqref{eq:trig_3} are well-defined. 

Second, consider the case where $S_2$ and $S_4$ are parallel and intersection (i.e., coincident).  Then, the corresponding path $C_1S_2C_3S_4C_5$ must be optimal. This result is established in the following lemma.
\begin{lemma}
\label{lem:special_case_p}
Consider the class of $C_1S_2C_3S_4C_5$ paths with fixed directions $L$ or $R$ for each curved segment. If there exists a $C_1S_2C_3S_4C_5$ path in which $S_2$ and $S_4$ are parallel and intersecting, then the path is the shortest path in the class.
\end{lemma} 
\begin{proof}
In a $C_1S_2C_3S_4C_5$ path where the lines  $S_2$ and $S_4$ are parallel and intersecting, the $C_3$ curve is zero. Therefore, $C_1S_2$ is the shortest path from $X_i$ to $\mathbf{x}_m$ in between all the paths of its type. The analysis for the path $S_4C_5$ is analogous. Therefore, the path consists of two shortest paths between $\mathbf{x}_m$, initial and final configurations. Therefore, the path is the shortest path in the given class of paths.
\end{proof}

Note that in the second case, the corrosponding path is the optimal of its type, therefore, no further analysis is required for finding the optimal heading at the midpoint.

\section{Three-point Dubins Algorithm}
\label{sec:Optimal CSCSC Paths}

In this section we propose a simple method to find the optimal path in the problem instance $(X_i, \mathbf{x}_m, X_f)$. First, leveraging the properties in Section \ref{sec:inversion_proof}, we propose a method to find an approximate midpoint heading.

\subsection{Approximation Method}
\label{sec:approx_method}
In this section we propose an approximation of the optimal heading at the mid-point $\mathbf{x}_m$. We assume that the pair-wise distances of $\mathbf{x}_i$, $\mathbf{x}_m$ and $\mathbf{x}_f$ go to infinity. Then, the length of segments $\overline{A\mathbf{x}_m}$ and $\overline{B\mathbf{x}_m}$ go to infinity which, by Equation \eqref{eq:inversion_eq}, implies  $|\overline{A'\mathbf{x}_m}|$ and $|\overline{B'\mathbf{x}_m}|$ approach zero.
From Lemma \ref{lem:min_rad} (see Appendix~\ref{appendix:1}), the radius of the circles $\mathrm{inv}(S_2, \mathcal{C})$ and $\mathrm{inv}(S_4, \mathcal{C})$ is bounded from below by $\frac{1}{4}R_{\min}$. Therefore, the angles  $\angle \mathbf{x}_mCA'$ and  $\angle \mathbf{x}_mDB'$ (see Figure \ref{fig:magnify_crop_triangles}) approach zero and the angles $\angle C\mathbf{x}_mA$ and  $\angle B\mathbf{x}_mD$ go to $\frac{\pi}{2}$.

Therefore, in terms of the angles $\beta_1 = \angle \mathbf{x}_mAB$, $\beta_2 = \angle \mathbf{x}_mBA$, $\theta = \angle \mathbf{x}_mCD$, we have
\[
\label{eq:init_guess_1}
\beta_1 - \alpha + \theta = \frac{\pi}{2}, \quad \beta_2 + \alpha + \theta = \frac{\pi}{2}.
\]

From these equations we can approximate the heading $\alpha$ at the mid-point $\alpha$ by 
\begin{equation}
	\label{eq:approx_heading}
\bar{\alpha} = \frac{\beta_1-\beta_2}{2}.
\end{equation}
The following result establishes the maximum error between $\bar \alpha$ and the optimal heading $\alpha^*$.  The proof is given in Appendix~\ref{appendix:1}.

\begin{proposition}[Maximum error of approximated heading] \label{prop:max_err_init_guess_1}
For the optimal path of Problem \ref{prob:three_point}, the following holds for the
optimal heading at $\mathbf{x}_m$:
\[
    \left|\alpha^* - \frac{\beta_1 - \beta_2}{2}\right| \leq \zeta.
\]
For the optimal path $P^*$, the bound $\zeta$ is defined as
\begin{enumerate}
\item $\zeta = 0$  if  $|\overline{A\mathbf{x}_m}|=|\overline{B\mathbf{x}_m}|$ and $P^*$ is $RSRSR$, $LSLSL$, $RSLSR$, or $LSRSL$;
\item $\zeta = \frac{\pi}{9}$  if $P^*$ is $RSRSR$ or $LSLSL$; 
\item $\zeta = \frac{\pi}{5}$  if $P^*$ is  $RSLSR$ or $LSRSL$; and
\item $\zeta = \frac{11\pi}{36}$ if $P^*$ is $RSRSL$, $LSRSR$, $RSLSL$, or $LSLSR$.
\end{enumerate}
\end{proposition}

Note that the $\zeta$ values in Proposition~\ref{prop:max_err_init_guess_1}  are the worst-case bounds, thus in Section \ref{sec:simulations}, we provide the mean deviation of the approximated heading from the optimal on $50000$ instances with various distance constraints. Also note that the worst-case bounds improve as the distance between the points increase.
Given these approximation of the optimal heading at the mid-point, we can initialize an
iterative method to converge to the optimal.

\subsection{Iterative Method}\label{sec:iterative_method}
Starting from the heading given in Section \ref{sec:inversion_proof} as the initial heading, we propose the following method for iteratively improving the heading. The  method converges to the
optimal heading by iteratively correcting the angle between the bisector of the two line segments of the path $C_1S_2C_3S_4C_5$
and the vector between the mid-point and the center of the circle associated with the
heading (see Figure \ref{fig:vecs}). Without loss of generality, assume that the center of the first curve is located at the
origin and the center of the final curve is located at $(x_f, 0)$ and let $\mathbf{x}_m = (x_m, y_m)$ be the
mid-point. We define  vectors $\vec{v_i}$ and $\vec{v_f}$ parallel to the first and second
line segments the path $C_1S_2C_3S_4C_5$. Let $\mathbf{x}_c = (x_c, y_c)$ be the center of the circle associated with a
heading at the mid-point. Let $\mathrm{Rot}_{\theta}$ be the rotation matrix with angle $\theta$ and $\vec{e}_v$ be the unit length vector in the direction of vector $\vec{v}$. We have, 
\[\vec{v}_i = \mathrm{Rot}_{\theta_i}[x_c, y_c], \quad \vec{v}_f = \mathrm{Rot}_{\theta_f}[x_c - x_f, y_c],\]
\[\vec{v}_m = [x_m - x_c, y_m - y_c], \quad \vec{v} = \vec{e}_{v} + \vec{e}_{v_f}.\]

\begin{figure}
\centering
\includegraphics[width=0.7\linewidth, keepaspectratio=true]{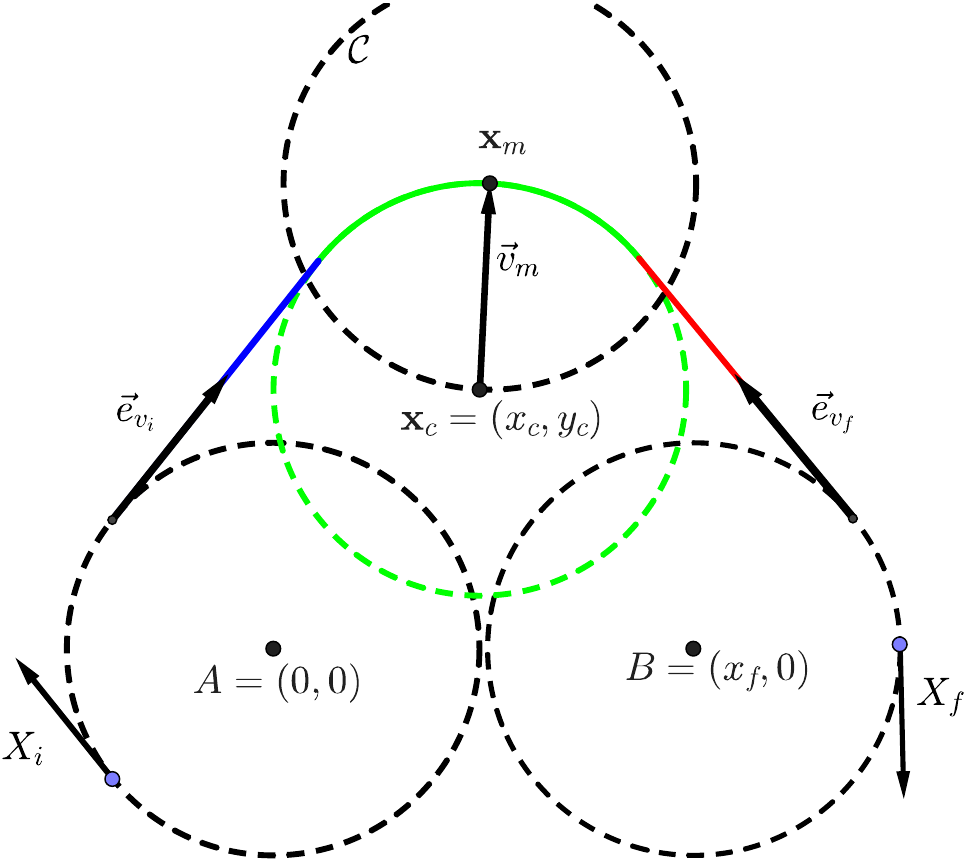}
\caption{The vectors $\vec{e}_{v_i}, \vec{e}_{v_f}$ and $\vec{v_m}$ for $R_1S_2R_3S_4R_5$ path.}
\label{fig:vecs}
\end{figure}

The angle $\theta_i$ is the angle of a common tangent of two circles $\mathrm{circle}(A, 1)$ and $\mathrm{circle}(\mathbf{x}_c, 1)$ from the line connecting the centers $A$ and $\mathbf{x}_c$. The angle $\theta_i$ equals zero if the line segment is an outer-common tangent and $\sin^{-1}(2/|v_i|)$, otherwise.
The algorithm for each path of type $C_1S_2C_3S_4C_5$ is as follows:
\begin{enumerate}
\item Find the approximated heading $\bar{\alpha}$ (Equation \eqref{eq:init_guess_1}),
\item Compute vectors $\vec{v_i}$, $\vec{v_f}$ and $\vec{v_m}$,
\item Compute the vector $\vec{v}$ bisecting the angle between $\vec{v_i}$ and~$\vec{v_f}$,
\item Return if vectors $\vec{v}$ and $\vec{v_m}$ are aligned,
\item compute the angle $\gamma$ between $\vec{v_i}$ and $\vec{v_m}$,
\item Rotate $(x_c, y_c)$ about $\mathbf{x}_m$ by $\gamma$,
\item continue from step (ii).
\end{enumerate}

The problem of finding the optimal heading at the mid-point is defined as the
following:
\begin{equation} 
\min_{x_c,y_c}  \cos^{-1}(\vec{e}_v \cdot \vec{v}_m)
\label{eq:obj}
\end{equation}

\begin{figure}
        \centering
        \begin{subfigure}{\linewidth}
        \centering
                \includegraphics[width=0.8\linewidth, keepaspectratio=true]{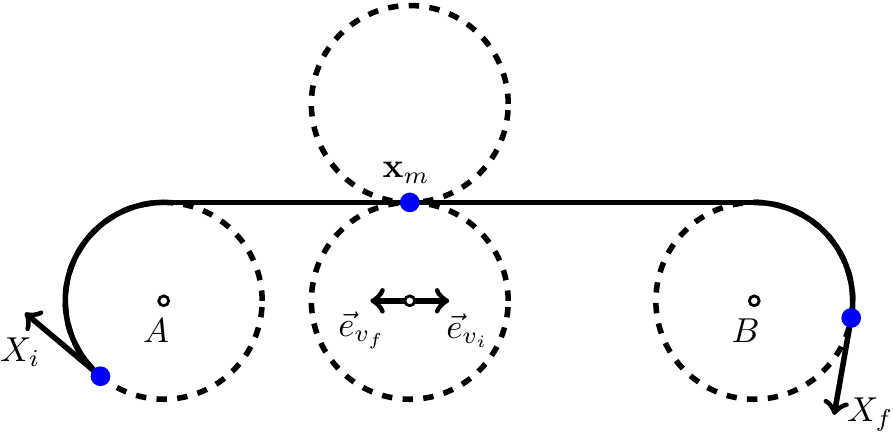}
                \label{fig:RS_1RS_2R}
        \end{subfigure}
        \begin{subfigure}{\linewidth}
        \centering
                \includegraphics[width=0.8\linewidth, keepaspectratio=true]{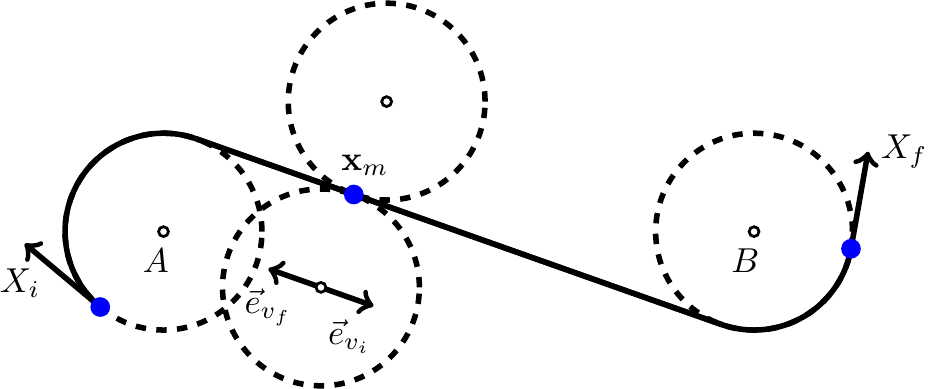}
        \label{fig:RS_1RS_2L}
        \end{subfigure}
        \caption{Two examples of the case that $\vec{e}_{v_i}$ and $\vec{e}_{v_f}$ sum to zero.}
        
        \label{fig:spec_cases}
\end{figure}

The minimum of~\eqref{eq:obj} occurs when the vectors $\vec{v}$ and $\vec{v_m}$ are parallel. Note that the derivative of the right hand side of objective function $\eqref{eq:obj}$ is not defined where the vectors $\vec{e}_v$
and $\vec{v}_m$ are parallel. However, minimizing~\eqref{eq:obj} is equivalent to the following maximization:
\begin{equation}
\label{eq:objective}
    \max_{x_c, y_c} \vec{e}_v \cdot \vec{v}_m
\end{equation}

To prove correctness of the iterative method, it suffices to show that all local maxima $(x_c,y_c)$ of~\eqref{eq:objective} are globally maximal.  The following lemma validates the iterative method.
\begin{lemma}
\label{lemma:convergence}
The center of the circle associated with the optimal heading is the unique maximizer of~\eqref{eq:objective}.
\end{lemma}

\begin{proof}
Note that the length of the vectors $\vec{v}_m$, $\vec{e}_{v_i}$ and $\vec{e}_{v_f}$ are independent of $\alpha$, and
the derivative of either of the vectors are orthogonal to the vector itself.
The derivative of $\vec{e}_v \cdot \vec{v}_m$ with respect to change of the heading $\alpha$ at the mid-point is as follows: 
\begin{align}
\begin{split}
    \frac{d(\vec{e}_v \cdot \vec{v}_m)}{d\alpha} =  - \frac{\vec{e}_{v_f}\cdot\frac{d\vec{e}_{v_i}}{d\alpha} + \vec{e}_{v_i}\cdot\frac{d\vec{e}_{v_f}}{d\alpha}}{|\vec{e}_{v_i} + \vec{e}_{v_f}|^3}(\vec{e}_{v_i} + \vec{e}_{v_f})\cdot \vec{v}_{m}  + \\ \frac{\frac{d\vec{e}_{v_i}}{d\alpha} + \frac{d\vec{e}_{v_f}}{d\alpha} }{|\vec{e}_{v_i} + \vec{e}_{v_f}|}\cdot v_{m} + \frac{\vec{e}_{v_i} + \vec{e}_{v_f}}{|\vec{e}_{v_i} + \vec{e}_{v_f}|} \cdot \frac{dv_{m}}{d\alpha}.
\end{split}
    \label{eq:deriv_1}
\end{align}

Case (i): The vectors $\vec{v}_i$ and $\vec{v}_f$ are dependent, Equation \eqref{eq:deriv_1} simplifies to
\begin{equation}\label{eq:deriv_2}
\frac{d(\vec{e}_v \cdot \vec{v}_m)}{d\alpha} =  \frac{\frac{d\vec{e}_{v_i}}{d\alpha} + \frac{d\vec{e}_{v_f}}{d\alpha} }{|\vec{e}_{v_i} + \vec{e}_{v_f}|}\cdot v_{m} + \frac{\vec{e}_{v_i} + \vec{e}_{v_f}}{|\vec{e}_{v_i} + \vec{e}_{v_f}|}\cdot \frac{dv_{m}}{d\alpha}.
\end{equation}

Note that $|\vec{e}_{v_i} + \vec{e}_{v_f}|$ equals zeros only if the vectors $\vec{e}_{v_i}$ and $\vec{e}_{v_f}$ are collinear with different directions. Therefore, the line segments are tangent to the circle at $\mathbf{x}_m$ and the heading 
corresponding to this case is the optimal heading (see Figure \ref{fig:spec_cases}). If $\vec{e}_{v_i}$ and $\vec{e}_{v_f}$ do not sum up to zero, then  
the optimal heading is the root of the following equation:
\begin{align*}
 (\frac{d\vec{e}_{v_i}}{d\alpha} + \frac{d\vec{e}_{v_f}}{d\alpha})\cdot v_{m} + (\vec{e}_{v_i} + \vec{e}_{v_f})\cdot \frac{dv_{m}}{d\alpha} \\= \frac{d}{d\alpha}((\vec{e}_{v_i} + \vec{e}_{v_f})\cdot \vec{v}_{m}) = 0.
 \end{align*}

Trivially, the roots of this equation occur where $\vec{v_m}$ is parallel to $\vec{e}_{v_i} + \vec{e}_{v_f}$.

Case (ii): The vectors $\vec{v_i}$ and $\vec{v_f}$ are linearly independent. Therefore, we can write $\vec{v_m}$
and the derivative as a linear combinations of $\vec{v_i}$ and $\vec{v_f}$, i.e.
\begin{equation}
\label{eq:linearly_independent1}
    \vec{v}_m = \frac{1}{|\vec{e}_{v_i} + \vec{e}_{v_f}|}\vec{e}_{v_i} + \frac{c}{|\vec{e}_{v_i} + \vec{e}_{v_f}|}\vec{e}_{v_f},
\end{equation}
\begin{equation}
\label{eq:linearly_independent2}
    \frac{d\vec{v}_m}{d\alpha} = \frac{1}{|\vec{e}_{v_i} + \vec{e}_{v_f}|}\mathrm{Rot}_{\frac{\pi}{2}}\vec{e}_{v_i} + \frac{c}{|\vec{e}_{v_i} + \vec{e}_{v_f}|}\mathrm{Rot}_{\frac{\pi}{2}}\vec{e}_{v_f}.
\end{equation}

In order to prove that the extremums occur where the bisector of the angle between $\vec{v_i}$ and $\vec{v_f}$ aligns with $\vec{v_m}$, we need to show that $c = 1$ is the only solution to Equation \eqref{eq:deriv_1}. 
Substituting Equations \eqref{eq:linearly_independent1} and \eqref{eq:linearly_independent2} into Equation \eqref{eq:deriv_1} and simplify the equation, we have

\begin{align*}(1 - c)(\vec{e}_{v_f} \cdot \mathrm{Rot}_{\frac{\pi}{2}}\vec{e}_{v_i} + c\vec{e}_{v_i} \cdot\mathrm{Rot}_{\frac{\pi}{2}}\vec{e}_{v_f}) = \\ (1-c)^2(\vec{e}_{v_i}\cdot \mathrm{Rot}_{\frac{\pi}{2}}\vec{e}_{v_f}) = 0.\end{align*}
With the assumption that $\vec{v_i}$ and $\vec{v_f}$ are linearly independent, the equation equals to zero if and
only if $c = 1$.
\end{proof}

An immediate consequence of Lemma \ref{lemma:convergence} is the convergence of the iterative method.
\begin{corollary}
The iterative method converges to the optimal heading at the mid-point.
\end{corollary}

\begin{remark}[Eliminating the distance constraint]
\label{rem:relaxing}
The $4R_{\min}$ distance constraint in Problem \ref{prob:three_point} ensures that the path types are of type $C_1S_2C_3S_4C_5$. Eliminating the distance constraint introduces additional path types, i.e., paths including $CC$ and $CCC$ segments. The proof of optimality for Lemma \ref{lemma:convergence} does not consider any distance constraint between the points. Therefore, the iterative method is applicable to any $C_1S_2C_3S_4C_5$ path type even when $4R_{\min}$ is not satisfied. Also, a constant time method is presented in \cite{isaiah2015motion} for computing $CC$ paths. Implementing the method for computing $CC$ paths alongside our iterative method for $C_1S_2C_3S_4C_5$ paths, we obtain a method to optimally find the heading at the mid point for all path types between three consecutive points with exception of $\{C_1C_2C_3S_4C_5, C_1S_2C_3C_4C_5, C_1C_2C_3C_4C_5\}$. Although these path types are not considered in our method, the extensive simulation results in Section \ref{sec:simulations} show that under the relaxed distance condition the paths generated by our method are in $0.1$ percent of the optimal path. \oprocend
\end{remark}

\section{Locally Optimizing a Dubins TSP Tour}
\label{sec:DubinsTSP}

The solution to Problem~\ref{prob:three_point} provides a method for locally optimizing a Dubins TSP tour in a post-processing phase. Given a set of $n$ points in the Euclidean plane, a solution to the Dubins TSP is an ordering of the $n$ points, along with a heading at each point that minimizes the total path length. Let $T$ be a Dubins tour such that $T_i$ is the $i$th configuration $(\mathbf{x}_i, \alpha_i)$.
Now we define our post-processing method as follows: 
\begin{enumerate}
\item For every $T_i$, solve the problem $(T_{i-1}, \mathbf{x}_i, T_{i + 1})$ and update $\alpha_i$,
\item Randomly delete a configuration $T_i$ in $T$ and re-insert to a position in the tour with minimum additional cost.
\end{enumerate}
Note that every segment of three consecutive vertices on the tour is a $(X_i, \mathbf{x}_m, X_f)$ problem instance. Therefore, in a tour of length $n$, finding the position to insert a point with minimum additional cost requires solving $n-1$ problem instances of type $(X_i, \mathbf{x}_m, X_f)$. 
The steps (i) and (ii) of refinements  terminates if there is no improvement in the path. 

\section{Simulation Results} 
\label{sec:simulations}
We evaluate the performance of the proposed approach on both randomly generated $(X_i, \mathbf{x}_m, X_f)$ instances and in post-processing Dubins TSP tours as in Section \ref{sec:DubinsTSP}. The point-to-point Dubins path \cite{shkel2001classification} and the three-point Dubins method are implemented in Python and the experiments are conducted on an Intel Corei5 @2.5Ghz processor. The experiments in this section consider a Dubins vehicle with $R_{\min} = 1$.
 
\subsection{Three-Point Dubins} \label{sec:random_instances}

In this section we compare performance of the initial approximation and the iterative method to discretizing the heading at $\mathbf{x}_m$ with 360 equally-spaced headings.  Let $\alpha_d$ be a heading among the discretized headings. The discretization method creates the configuration $X_m = (\mathbf{x}_m, \alpha_d)$, and solves two Dubins path problems, namely $(X_i, X_m)$ and $(X_m, X_f)$. The discretization method returns the minimum path among the headings. 

Figure \ref{fig:distrib1} (top) shows the percentage deviation in path length for the approximate heading $\bar \alpha$ in~\eqref{eq:approx_heading}, relative to the path length computed using $360$ heading discretizations. The bottom figure shows the deviation of the path length produced by the iterative method to that of the discretized heading. The experiments are conducted on $50000$ random $(X_i, \mathbf{x}_m, X_f)$ instances, where the points are uniformly randomly selected in a $10 \times 10$ environment. The $x$-axis in Figure \ref{fig:distrib1} is the rounded minimum distance of the three points. For example, $1$ on the x-axis represent the instances where the minimum distance between the points is in interval $[1, 1.5)$.
The negative values represent instances in which the proposed methods outperform the discretization method. 
The distribution shows that even in the cases where points are less than $4R_{\min}$ apart from each other, the iterative method generates shorter paths. 

The average computation time may vary based on the distances of the points due to considering additional path types mentioned in Remark \ref{rem:relaxing}. The iterative method improves the runtime of computing a three-point Dubins path, under $4R_{\min}$ distance constraint, compared to $360$ discretization by a factor of $13.65$. However, this factor of improvement is $5.21$ for the instances with points less than $2R_{\min}$ apart.  Table \ref{time_ratios} shows the factor of improvement in runtime of the iterative and approximation method when compared to discretization with $360$ headings.         

\begin{figure}
	\begin{subfigure}{\linewidth}
	\includegraphics[width=\linewidth, keepaspectratio=true]{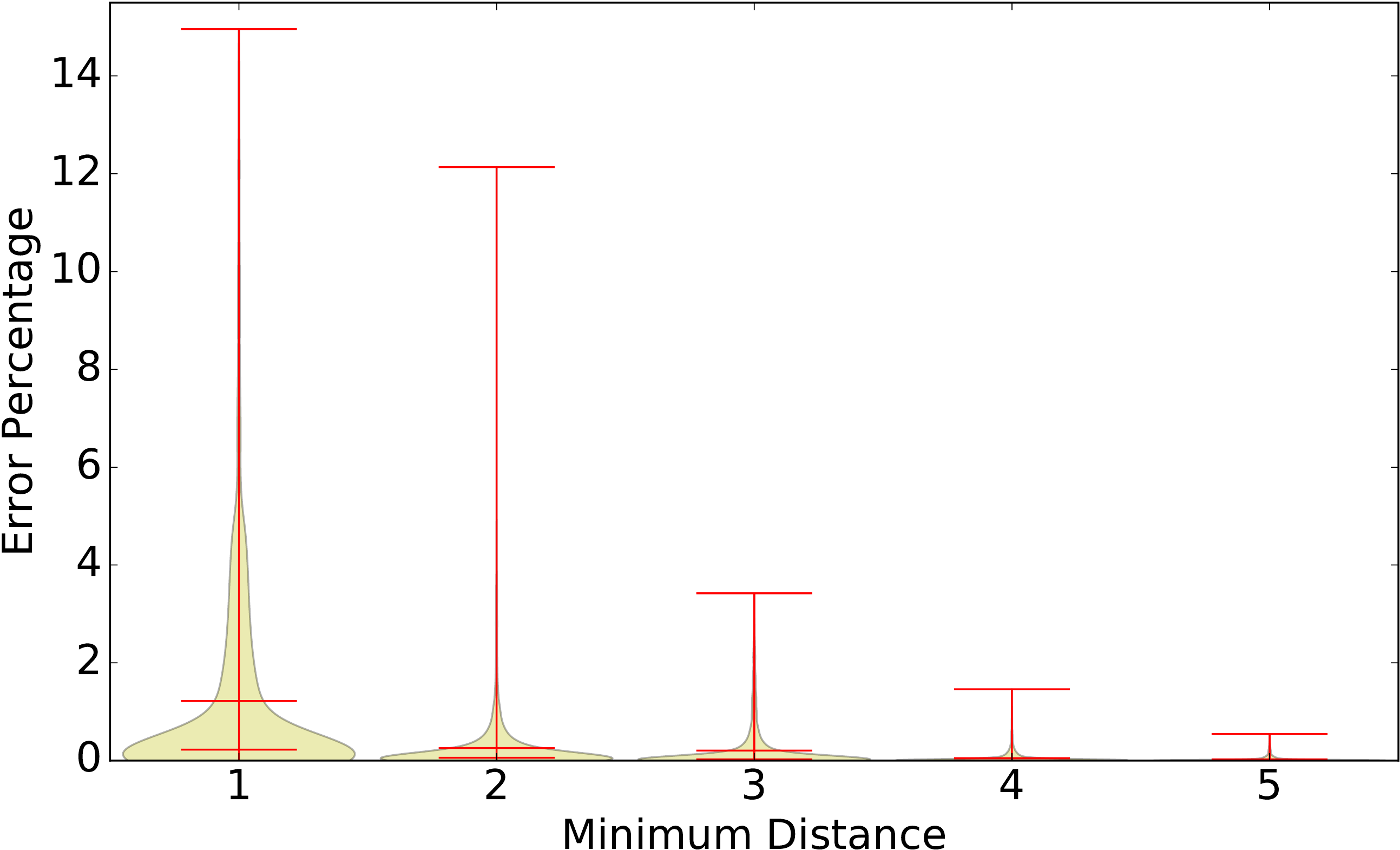}
	\label{fig:distrib2}
	\end{subfigure}
	\begin{subfigure}{\linewidth}
	\includegraphics[width=\linewidth, keepaspectratio=true]{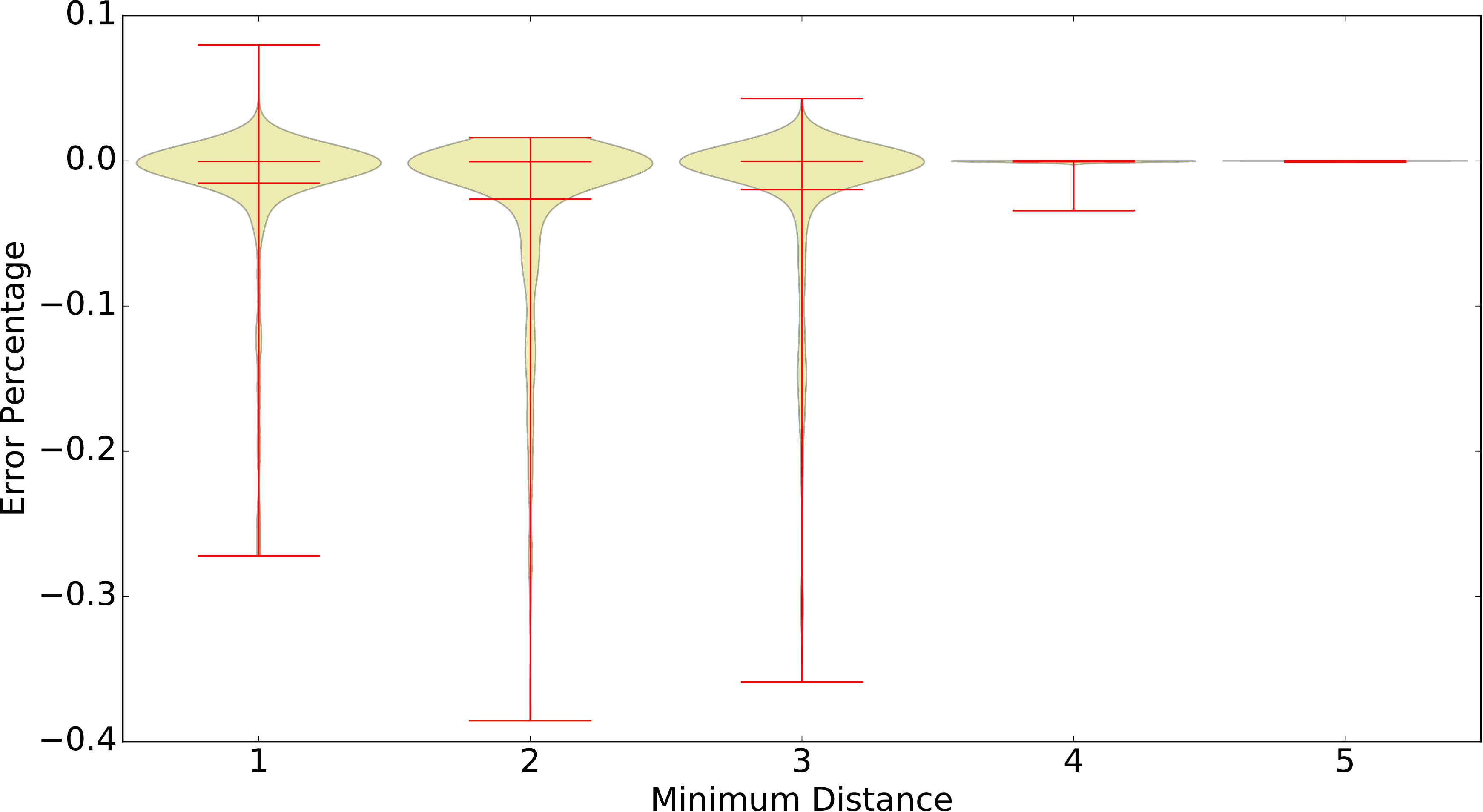}
	\label{fig:distrib}
	\end{subfigure}
	\caption{The percentage deviation of the length of paths  generated by the approximation method (top) and iterative method (bottom) from the discretization method with $360$ equally spaced headings. The width of the distributions represent the probability of occurring an instance in the corresponding difference percentile.} 
	\label{fig:distrib1}
\end{figure}

\begin{table} 
\small
\centering
\begin{tabular}{@{} lrrr @{}}
\toprule
&\multicolumn{1}{c}{$2R_{\min}$} &\multicolumn{1}{c}{$3R_{\min}$} &\multicolumn{1}{c}{$4R_{\min}$} \\
\midrule
Approx. heading 				&  $65.3$ 	& $67.1$ 		& $74.2$	\\
Iterative method  				&  	$5.2$ 	& $6.8$  		& $13.6$\\
\bottomrule
\end{tabular}
\caption{The factor of improvement in runtime of the iterative and approximation method over $360$ discrete headings. The average solver time on $10000$ instances for the discretization method with $360$ equally spaced headings is $0.01728$ seconds.}
\label{time_ratios}
\end{table}

\subsection{Post-processing on Dubins Tour} 
\label{sec:random_Dubins}
In this experiment, we implement the GTSP method \cite{le2012dubins} on random instances with various discretization levels followed by our post-processing method in Section \ref{sec:DubinsTSP}. Given a Dubins TSP on $n$ points, and a discretization level of $d$ at each point, the GTSP instance will have $nd$ vertices. The results show the advantages of the local optimization on GTSP solutions with coarse discretization over solving GTSP with fine discretization. 

To characterize the performance of our algorithm, we conduct experiments on low and high-density Dubins TSP instances. Table~\ref{tab:high_dense} shows the results on uniformly randomly generated instances.  Each row of the table is a class of $20$ random instances with the same problem parameters: that is, the environment size $W$ $\times$ $W$, the number of points $N$, and the minimum pair-wise distance $D$.

The GTSP instances are solved using the state-of-the-art GTSP solver, GLKH~\cite{helsgaun2015solving} which is implemented in $\mathtt{C}$. In Table \ref{tab:high_dense} the abbreviations G.\ Len and G.\ Time represent the average tour length and solver time, in seconds, for the GTSP solver. Similarly, P.\ Len represents the average tour length after post-processing and P.\ Time represents the time required for the post-processing of the GTSP tour. The total time of the GTSP approach and the post-processing is denoted by Time.
Table \ref{tab:high_dense} (top) shows the performance of the post-processing technique on the GTSP tours with a discretization level of $1$ and $10$ in low-density Dubins TSP instances. The time and the tour length of the GTSP solution with discretization level $20$ is the reference, denoted by ref, for evaluating the performance of the post-processing method. The table (top) includes the ratios of the total time and post-processed tour length to the reference. In the class of instances N30W20D2.0, the deviation of the post-processed tour length from the reference is $3.7\%$ and the total time of solving the GTSP with $1$-discretization followed by the post-processing technique is just $1\%$ of the solver-time of the GTSP approach with discretization level $20$.  

In an environment with high density of points, the discretization level has larger impact on the ordering of the points in a GTSP solution.  Table~\ref{tab:high_dense} (bottom) shows the results of the GTSP tour with post-processing on high-density instances. For example, the results on the class of instances N50W20D0.0 show that the tour length of the post-processed GTSP tour with discretization level $5$ is $5.3\%$ longer than the GTSP tour with discretization level $20$.  However, the runtime is improved by a factor of $13.26$.

\begin{table*}[htb] 
\centering
\begin{tabular}{@{} cl rrrrrrrrrrrrr @{}}
\toprule
 & & P. Time & \multicolumn{5}{c}{GTSP 1-discretization} & \vphantom{a} &\multicolumn{5}{c}{GTSP 10-discretization} & \vphantom{a} \\
 \cmidrule{4-8} \cmidrule{10-14}
& &  & G. Time & G. Len & P. Len & $\text{P. Len/ref}$ & $ \text{Time/ref}$ && G. Time & G. Len & P. Len & $\text{P. Len/ref}$ & $\text{Time/ref}$
\\

\midrule
\multirow{7}{*}{\rotatebox{90}{Low Density}} & N10W15D4.0 	& 	0.1 	& 	0.0 	&  	75.5   & 54.6 & 	1.000 & 	0.02 && 0.8	& 	54.8 & 54.6 	& 	1.000 & 	0.17 \\
									& N10W10D3.0 	& 	0.1 	& 	0.0 	&  	66.0   & 38.2 & 1.003	& 	0.03	 && 0.6		& 	38.4 		& 	38.1 	& 1.000	& 	0.18\\
									& N20W20D3.0 	& 	0.3 	& 	0.0 	&  	142.5   	& 94.5 &  1.002	& 	0.01	&& 9.5	& 	94.7 	&  94.5 & 1.002 & 0.26 \\
									& N30W20D2.0 	& 	0.5 	& 	0.4 	&  	187.3   & 110.3 &  1.037	& 	0.01	 && 18.4	& 	107.1 & 106.4 	& 	1.000 & 	0.18 \\
									& N30W30D3.0 	& 	0.7 	& 	0.1 	&  	229.0   	& 157.6 &  1.006	& 	0.01 	&& 20.5	& 	157.1 	&  156.7 	& 	1.000 &  0.15	\\
									& N40W30D4.0 	& 	1.6 	& 	0.3 	&  	289.7   & 205.3 & 1.022 & 	0.01	&&  45.8	& 	201.3 & 200.7 	& 	1.000 & 	0.16 \\
\bottomrule
\end{tabular}
\vskip1em
\begin{tabular}{@{} cl rrrrrrrrrrrrr @{}}
\toprule
& & P. Time & \multicolumn{3}{c}{GTSP 5-discretization} & \vphantom{a} &\multicolumn{3}{c}{GTSP 10-discretization} & \vphantom{a} & \multicolumn{3}{c}{GTSP 20-discretization}\\
 \cmidrule{4-6} \cmidrule{8-10} \cmidrule{12-14}
& &  & G. Time & G. Len & P. Len && G. Time & G. Len & P. Len && G. Time & G. Len & P. Len
\\
\midrule
\multirow{7}{*}{\rotatebox{90}{High Density}}& N10W5D0.0 	& 	0.3 	& 	1.5 	&  	28.0   & 25.1 && 2.4	& 	24.6 & 22.9 	&& 6.5 & 21.4 & 21.2	&\\
										& N20W5D0.0 	& 	0.6 	& 	7.5 	&  	44.6   & 41.8 && 16.5	& 	38.6 & 36.4 && 46.2 & 35.1	& 34.8\\
										& N30W5D0.0 	& 	2.3 	& 	24.5 	&  	58.2   & 54.5 && 39.2	& 	52.5 & 50.9 && 121.4 & 46.6 & 46.0\\
										& N30W20D0.0 	& 	0.7 	& 	3.6 	&  	103.0   & 97.5 && 16.0	& 	97.0 & 93.2 && 67.3	& 	95.4& 92.3\\
										& N40W5D0.0 	& 	2.1 	& 	39.0 	&  	71.8   & 70.2 && 77.9	& 	65.5 & 63.2 && 222.4 & 61.4  & 61.1&\\
										& N50W20D0.0 	& 	0.8 	& 	23.2 	&  	153.5   & 145.2 &&  63.3	& 	141.7 & 137.9&& 318.0  & 137.7 & 136.9\\
\bottomrule
\end{tabular}
\caption{Average tour length and time of the GTSP approach compared to the post-processing method on random instances with low-density of points (top table) and high-density (bottom table). The instance names consist of the environment size $W$ $\times$ $W$, the number of points $N$, and the minimum pair-wise distance $D$.
}
\label{tab:high_dense}
\end{table*}

\section{Conclusion and Future Work}
This paper considers the optimal Dubins path through three consecutive points.  The presented approximation method followed by the iterative method show improvement in run-time compared to the discretization of the headings. In addition to the experimental results, the application of inversive geometry provides a direction for further research. In addition, we hope to extend the analysis to path types containing $CCC$ segments, as described in Remark~\ref{rem:relaxing}.

\appendices
\section{Proof of Results} 
\label{appendix:1}

To prove Proposition \ref{prop:max_err_init_guess_1} we require the following result.
\begin{lemma}[Minimum radius of an inverted circle] \label{lem:min_rad}
The radius of the inverted  circles $\mathrm{circle}(C, R)$ and $\mathrm{circle}(D, R)$ in the optimal path are greater than $\frac{R_{\min}}{4}$.
\end{lemma}
\begin{proof}
Referring to Figure~\ref{fig:inversion_proof}, the line segments $S_2$ and $S_4$ are tangent to $\mathrm{circle}(\mathbf{x}_c, R_{\min})$.  The inverse of $S_2$ and $S_4$
with respect to $\mathcal{C} = \mathrm{circle}(\mathbf{x}_m, R_{\min})$ are the circles $\mathrm{circle}(C, R)$ and $\mathrm{circle}(D, R)$, respectively.  For any point $P$ in $S_2$ ($S_4$) the distance $|\overline{\mathbf{x}_mP}|$ is inversely proportional to $|\overline{\mathrm{inv}(P, \mathcal{C})\mathbf{x}_m}|$. Note that the closest point to $\mathbf{x}_m$ on the line segments $S_2$ and $S_4$, results in the farthest point from $\mathbf{x}_m$ in $\mathrm{circle}(C, R)$ and $\mathrm{circle}(D, R)$. Therefore, the radius $R$ is minimum if the minimum distance of the point $\mathbf{x}_m$ from either of the line segments $S_2$ or $S_4$ is maximum.  Since the point $\mathbf{x}_m$ lies on$\mathrm{circle}(\mathbf{x}_c, R_{\min})$, by the triangle inequality, the minimum distance of the point $\mathbf{x}_m$ from the line segments $S_2$ and $S_4$ is at most $2R_{\min}$.  Thus, the farthest distance of $\mathbf{x}_m$ from a point on $\mathrm{circle}(C, R)$ and $\mathrm{circle}(D, R)$ is at most $\frac{1}{2}R_{\min}$ which implies the radius $R$ is at least $\frac{1}{4}R_{\min}$.
\end{proof} 

\begin{proof}[Proof of Proposition \ref{prop:max_err_init_guess_1}]

For a path contained in $\{RSRSR, LSLSL, RSLSR, LSRSL\}$, 
Equations \eqref{eq:trig_1}-\eqref{eq:trig_3} simplify to the following:
\begin{equation} \label{eq:simply_equtions}
|\overline{A\mathbf{x}_m}| \cos(\beta_1 - \theta + \alpha) = |\overline{B\mathbf{x}_m}| \cos(\beta_2 + \theta + \alpha).
\end{equation}
With the assumption of case (i) ($|\overline{A\mathbf{x}_m}| = |\overline{B\mathbf{x}_m}|$) and by Equation \eqref{eq:simply_equtions}, the optimal
heading $α$ is equal to the approximated heading, namely $\alpha = \frac{\beta_1 - \beta_2}{2}$. 

Case (ii): Under the distance constraint of Problem \ref{prob:three_point} ($4R_{\min}$), the approximated heading has maximum deviation from the optimal heading when the radius $R$ reaches infinity or its lower bound from Lemma \ref{lem:min_rad}, namely $\frac{1}{4}$. For the
path types $RSRSR$ and $LSLSL$, substituting the minimum and maximum $R$ values in the
Equations \eqref{eq:trig_1} and $\eqref{eq:trig_2}$ result the following:
\begin{align*}
-\alpha + \beta_1 + \theta &\in \left[\cos^{-1}\Big(\frac{1}{3}\Big), \cos^{-1}\Big(\frac{-1}{3}\Big)\right], \\
\alpha + \beta_2 + \theta &\in \left[\cos^{-1}\Big(\frac{1}{3}\Big), \cos^{-1}\Big(\frac{-1}{3}\Big)\right].
\end{align*}
Therefore, the maximum deviation of the approximated heading from the optimal heading is in
the range $[\frac{-\pi}{9} , \frac{\pi}{9}]$. 

Case (iii): using similar argument:
\begin{align*}
-\alpha + \beta_1 + \theta &\in \left[\cos^{-1}(1), \cos^{-1}\Big(\frac{1}{3}\Big)\right],\\
\alpha + \beta_2 + \theta &\in \left[\cos^{-1}(1), \cos^{-1}\Big(\frac{1}{3}\Big)\right].
\end{align*}
Therefore, the optimal heading lies within the interval $[-\frac{\pi}{5}, \frac{\pi}{5}]$ centered at the approximated heading.

Case (iv): We prove the bound for path types in $\{RSRSL,LSLSR\}$ and the other path types directly follow same analysis. In these types of paths, substituting the upper and lower bounds of $R$ in Equations \eqref{eq:trig_1}-\eqref{eq:trig_3} yields the following:
\[\frac{-1}{3}\leq \cos(\beta_1 + \theta - \alpha) \leq \frac{1}{3},\]
\[\frac{1}{3}\leq \cos(\beta_2 + \theta + \alpha) \leq 1.\]
Thus, 
\[\beta_1 + \theta - \alpha \in \left[\frac{\pi}{2} - \frac{\pi}{9} ,\frac{\pi}{2} + \frac{\pi}{9}\right],
\quad
\beta_2 + \theta + \alpha \in \left[0 ,\frac{\pi}{2} - \frac{\pi}{9}\right].\]
Therefore, the maximum error of $\bar \alpha$ is $\frac{11\pi}{36}$.
\end{proof}
 
\bibliographystyle{IEEEtran}


\end{document}